\documentclass[11pt]{amsart}
\usepackage[margin=3cm]{geometry}
\usepackage{enumerate}
\usepackage{amsmath}
\usepackage{amssymb,latexsym}
\usepackage{amsthm}
\usepackage{color}
\usepackage[dvipsnames]{xcolor}
\usepackage{cancel}
\usepackage{graphicx}
\usepackage{cite}
\usepackage{changes}

\makeatletter
\renewcommand{\tocsection}[3]{%
  \indentlabel{\@ifnotempty{#2}{\bfseries\ignorespaces#1 #2\quad}}\bfseries#3}
\renewcommand{\tocsubsection}[3]{%
  \indentlabel{\@ifnotempty{#2}{\ignorespaces#1 #2\quad}}#3}

\def\@tocline#1#2#3#4#5#6#7{\relax
  \ifnum #1>\c@tocdepth 
  \else
    \par \addpenalty\@secpenalty\addvspace{#2}%
    \begingroup \hyphenpenalty\@M
    \@ifempty{#4}{%
      \@tempdima\csname r@tocindent\number#1\endcsname\relax
    }{%
      \@tempdima#4\relax
    }%
    \parindent\z@ \leftskip#3\relax \advance\leftskip\@tempdima\relax
    \rightskip\@pnumwidth plus1em \parfillskip-\@pnumwidth
    #5\leavevmode\hskip-\@tempdima{#6}\nobreak
    \leaders\hbox{$\m@th\mkern \@dotsep mu\hbox{.}\mkern \@dotsep mu$}\hfill
    \nobreak
    \hbox to\@pnumwidth{\@tocpagenum{\ifnum#1=1\bfseries\fi#7}}\par
    \nobreak
    \endgroup
  \fi}
\AtBeginDocument{%
\expandafter\renewcommand\csname r@tocindent0\endcsname{0pt}
}
\def\l@subsection{\@tocline{2}{0pt}{2.5pc}{5pc}{}}
\makeatother

\usepackage{etoolbox}
\makeatletter
\patchcmd{\@setaddresses}{\indent}{\noindent}{}{}
\patchcmd{\@setaddresses}{\indent}{\noindent}{}{}
\patchcmd{\@setaddresses}{\indent}{\noindent}{}{}
\patchcmd{\@setaddresses}{\indent}{\noindent}{}{}
\makeatother

\newcommand{\qqbin}[2]{\begin{bmatrix}
		{#1} \\ {#2 }
\end{bmatrix}_q}

\usepackage[linktocpage]{hyperref}
\hypersetup{
  colorlinks   = true, 
  urlcolor     = blue, 
  linkcolor    = Purple, 
  citecolor   = red 
}
\usepackage{comment}
\usepackage{amscd}
\usepackage{mathtools}

\DeclareMathOperator{\C}{\mathcal{C}}
\newcommand{\srk}{\mathrm{srk}}
\newcommand{\dsrk}{\mathrm{d}_{\mathrm{srk}}}

\DeclareMathOperator{\rk}{rk}

\DeclareMathOperator{\Mat}{Mat}

\DeclareMathOperator{\ww}{w}
\usepackage{cleveref}

\theoremstyle{definition}
\newtheorem{theorem}{Theorem}[section]

\newtheorem{lemma}[theorem]{Lemma}
\newtheorem{corollary}[theorem]{Corollary}
\newtheorem{definition}[theorem]{Definition}
\newtheorem{proposition}[theorem]{Proposition}

\newtheorem{remark}[theorem]{Remark}

\newcommand{\F}{{\mathbb F}}

\newcommand{\NN}{{\mathbb N}}

\newcommand{\bfn}{\mathbf {n}}
\newcommand{\bfm}{\mathbf {m}}

\newcommand{\fq}{{\mathbb F}_{q}}
\newcommand{\Fq}{{\mathbb F}_{q}}
\newcommand{\Fm}{{\mathbb F}_{q^m}}

\newcommand{\spacmn}{\Mat(\bfn,\bfm,\F_q)}

\newcommand{\st}{\,:\,}

\usepackage{todonotes}

\newcommand{\qbinom}[2]{\genfrac{[}{]}{0pt}{}{#1}{#2}_q}

\title{On the non-existence of perfect codes in the sum-rank metric}
\date{}

\author[G. Del Prete]{Giuseppe Del Prete}
\address{Giuseppe Del Prete, \textnormal{Dipartimento di Matematica e Fisica, Universit\`a degli Studi della Campania ``Luigi Vanvitelli'', Viale Lincoln, 5, I--\,81100 Caserta, Italy}}
\email{giuseppe.delprete8@studenti.unicampania.it}

\author[A. Roccolano]{Antonio Roccolano}
\address{Antonio Roccolano, \textnormal{Dipartimento di Matematica e Fisica, Universit\`a degli Studi della Campania ``Luigi Vanvitelli'', Viale Lincoln, 5, I--\,81100 Caserta, Italy}}
\email{antonio.roccolano1@studenti.unicampania.it}

\author[F. Zullo]{Ferdinando Zullo}
\address{Ferdinando Zullo, \textnormal{Dipartimento di Matematica e Fisica, Universit\`a degli Studi della Campania ``Luigi Vanvitelli'', Viale Lincoln, 5, I--\,81100 Caserta, Italy}}
\email{ferdinando.zullo@unicampania.it}

\subjclass[2020]{11T71; 51E20; 94B05} 
\keywords{Sum-rank metric; perfect code; Sphere-packing bound}
\begin{document}

\begin{abstract}
We study perfect codes in the sum-rank metric, a generalization of both the Hamming and rank metrics relevant in multishot network coding and space-time coding. A perfect code attains equality in the sphere-packing bound, corresponding to a partition of the ambient space into disjoint metric balls. While perfect codes in the Hamming and rank metrics are completely classified, the existence of nontrivial perfect codes in the sum-rank metric remains largely open.
In this paper, we investigate linear perfect codes in the sum-rank metric. We analyze the geometry of balls and derive bounds on their volumes, showing how the sphere-packing bound applies. For two-block spaces, we determine explicit parameter constraints for the existence of perfect codes. For multiple-block spaces, we establish non-existence results for various ranges of minimum distance, divisibility conditions, and code dimensions. We further provide computational evidence based on congruence conditions imposed by the volume of metric balls.
\end{abstract}

\maketitle


\section{Introduction}

Consider a finite metric space $(X,d)$.  
An \emph{$s$-code} $\mathcal{C}$ in $X$ is a subset of $X$ with at least two elements such that for every $c_1, c_2 \in \mathcal{C}$ with $c_1 \ne c_2$, we have $d(c_1,c_2) \geq s$.

For $c \in X$ and an integer $r \geq 0$, we define the ball of radius $r$ centered at $c$ as
\[
B_r(c) = \{ a \in X : d(a,c) \leq r \}.
\]

Since $d$ is a metric, one immediately observes that the balls of radius  
\[
k = \Big\lfloor \tfrac{s-1}{2} \Big\rfloor
\]
centered at the elements of an $s$-code are pairwise disjoint, and their union is contained in $X$. This yields the bound
\[
\sum_{c \in \mathcal{C}} |B_k(c)| \leq |X|.
\]

If, in addition, the metric $d$ is translation-invariant, then all balls of radius $k$ have the same size, so
\begin{equation}\label{eq:spherepackinggen}
|\mathcal{C}| \, |B_k(c)| \leq |X|, \qquad \text{for some } c \in \mathcal{C}.
\end{equation}

When $d$ is the Hamming metric and $X=\mathbb{F}_q^n$, where $\fq$ is a finite field of order $q$, the bound \eqref{eq:spherepackinggen} is known as the \emph{sphere-packing bound} (or \emph{Hamming bound}), and codes achieving equality in \eqref{eq:spherepackinggen} are called \emph{perfect codes}.  

By analogy, we will refer to \eqref{eq:spherepackinggen} as the \emph{sphere-packing bound} in the general setting, and to codes attaining equality as \emph{perfect codes}.

Examples of linear perfect codes in the Hamming metric are:
\begin{itemize}
    \item trivial codes (i.e., $\mathbb{F}_q^n$ and the repetition code);
    \item Hamming codes;
    \item the binary and ternary Golay codes of length $23$ and $11$, respectively.
\end{itemize}
After a series of works (see \cite{cohen1964note,alter1968nonexistence,alter1968non,van1971nonexistence,van1970nonexistence,van1969nonexistence,shapiro1959mathematical,tietavainen1956there}), 
Tiet{\"a}v{\"a}inen \cite{tietavainen1973nonexistence} completely classified the linear perfect codes in the Hamming metric as precisely those listed above (up to equivalence).

Another well-studied case is the rank metric. 
Taking $X=\mathbb{F}_q^{m\times n}$ with $d$ the rank metric, Loidreau \cite{loidreau2006properties} proved that the only linear perfect codes in this setting are the trivial ones.  
More recently, the situation in which $X$ is a restricted matrix space (i.e., when $X$ is the space of symmetric, Hermitian, or alternating matrices over $\mathbb{F}_q$) has also been studied; see \cite{mushrraf2025perfect,mushrraf2024perfect,abiad2025eigenvalue}.

This paper is dedicated to the study of perfect codes in the sum-rank metric.  
Let $\mathbb{F}_q$ be a finite field, and consider a vector partitioned into several blocks, where each block is a matrix over $\mathbb{F}_q$. The \emph{sum-rank weight} of such a vector is defined as the sum of the ranks of all its blocks. Similarly, the \emph{sum-rank distance} between two vectors is defined as the sum of the ranks of the differences between corresponding blocks.  

The sum-rank metric generalizes both the Hamming and the rank metrics. If each block consists of a single element, it coincides with the Hamming metric; if there is only one block, it reduces to the rank metric. This metric is particularly useful in coding theory, as it captures both block-wise and rank-based error structures. Applications include \emph{multishot network coding}, where errors can affect multiple transmissions, and \emph{space-time coding} in wireless communications, where blocks correspond to different antennas or time slots. For further details, we refer to \cite{byrne2021fundamental,martinez2022codes,martinez2019universal}.  

Covering properties have been widely studied in both the rank-metric and the sum-rank metric (see \cite{bartoli2024saturating,bonini2023saturating,bonini2024geometry,zullo2024saturating}), and a detailed study of the geometry of balls in the sum-rank metric has been carried out in \cite{sauerbier2025bounds,puchinger2022generic,byrne2021fundamental,ott2021bounds}.  

The first study of perfect codes in the sum-rank metric was conducted by Martinez-Pe\~nas in \cite{martinezpenas2021hamming}, where he observed that perfect codes in the Hamming metric can also be viewed as perfect codes in the sum-rank metric. He left as an open question the existence of perfect codes in the sum-rank metric that do not arise from the Hamming metric. To the best of our knowledge, no other perfect codes are known, though a construction of \emph{quasi-perfect sum-rank metric codes} was recently provided by Chen in \cite{chen2024quasi}.  

In this paper, we focus on detecting parameter sets for which linear perfect codes in the sum-rank metric do not exist.  
After recalling some preliminaries and notation for sum-rank metric codes in Section \ref{sec:prel}, in Section \ref{sec:balls} we study spheres and balls in the sum-rank metric and describe how the sphere-packing bound appears in this context. We also derive some bounds on the volumes of spheres and balls.  

In Section \ref{sec:t=2}, we study perfect codes in $\mathbb{F}_q^{\,n\times n} \oplus \mathbb{F}_q^{\,n\times n}$. Following the approach of Loidreau in \cite{loidreau2006properties}, which combines the bounds we have on the volumes of the balls with the Singleton-like bound, we prove that perfect codes in $\mathbb{F}_q^{n\times n} \oplus \mathbb{F}_q^{n\times n}$ can exist only for specific parameters.  

Section \ref{sec:t>2} addresses the case of perfect codes in $\underbrace{\mathbb{F}_q^{n\times n}\oplus \ldots \oplus \F_q^{n\times n}}_{t \text{ times}}$. The situation here is more complicated, but we are able to provide non-existence conditions when:  
\begin{itemize}
    \item the minimum distance is $3$ or $4$ and $q$ is sufficiently large (see Proposition \ref{prop:d-1/2=1});  
    \item the minimum distance is $5$ or $6$ and $t$ satisfies certain divisibility conditions (see Proposition \ref{prop:d-1/2=2});  
    \item the minimum distance is $5$ or $6$ and the dimension of the code is not a multiple of $n$ (see Proposition \ref{prop:n|dim});  
    \item the minimum distance is large enough with respect to $t$ (see Theorems \ref{thm:dlarg1} and \ref{thm:dlarg2});  
    \item $q$, $t$, and the minimum distance satisfy certain divisibility conditions (see Theorem \ref{thm:conddivis1} and Proposition \ref{prop:divisvcond2}).  
\end{itemize}  

In Section \ref{sec:comput}, since the existence of perfect codes in $\underbrace{\mathbb{F}_q^{n\times n}\oplus \ldots \oplus \F_q^{n\times n}}_{t \text{ times}}$ requires that the volume of a ball of radius $\lfloor \frac{d-1}{2} \rfloor$ (where $d$ is the minimum distance of the code) satisfies a congruential relation modulo the characteristic of the field, we provide some computational results based on this observation.

\section{Sum-rank metric codes}\label{sec:prel}

We fix now the notation that we will use for the whole paper. For us $q$ is a prime power and $\Fq$ is the finite field with $q$ elements. We will often consider the degree $m$ extension field $\Fm$ of $\Fq$. 
\noindent Let $t$ be a positive integer. From now on
$\bfn=(n_1,\ldots,n_t), \bfm=(m_1,\ldots,m_t) \in \NN^t$ will always be ordered tuples with $n_1 \geq n_2 \geq \ldots \geq n_t\geq 2$ and $m_1 \geq m_2 \geq \ldots \geq m_t \geq 2$.
We use the following compact notations for the direct sum of vector spaces of matrices
$$ \spacmn\coloneqq \bigoplus_{i=1}^t \F_q^{n_i \times m_i}.$$

In this section we recall the basic notions of sum-rank-metric codes seen as elements in $\spacmn$, that will be useful for the rest of the paper.  

\begin{definition}
Let $X\coloneqq(X_1,\dots, X_t)\in\Mat(\mathbf{n},\mathbf{m},\Fq)$. 
The \textbf{sum-rank weight} of $X$ is the quantity
$$\ww_{\srk}(X)\coloneqq\sum_{i=1}^t \rk(X_i).$$
\end{definition}

With these definition in mind, we can endow the space $\spacmn$ with a distance function, called the \textbf{sum-rank distance},
\[
\dsrk : \spacmn \times \spacmn \longrightarrow \mathbb{N}  
\]
defined by
\[
\dsrk(X,Y) \coloneqq \ww_{\srk}(X-Y).
\]

In particular, if $t=1$ the sum-rank distance corresponds to the rank distance, and if all the $n_i$'s or all the $m_i$'s are equal to one, then the sum-rank distance coincides with the Hamming distance.

We can now give the definition of sum-rank metric codes.

\begin{definition}
 A \textbf{sum-rank metric code} $\mathcal{C}$ is an $\F_q$-linear subspace of $\spacmn$ endowed with the sum-rank distance. 
The \textbf{minimum sum-rank distance} of a sum-rank code $\mathcal{C}$ is defined as usual via $$\dsrk(\mathcal{C})\coloneqq\min\{\ww_{\srk}(X) \st X \in \mathcal{C}, X \neq \mathbf{0}\}.$$ 
\end{definition}


The parameters of a sum-rank metric code are related via a Singleton-like bound, as proved in \cite[Theorem 3.2]{byrne2021fundamental}.

\begin{theorem}  \label{th:SingletonboundmatrixRav} \textbf{(Singleton-like bound)} \\
    Let $\mathcal{C}$ be a sum-rank metric code in $\Mat(\textbf{m},\bfn,\F_q)$ with $n_i \leq m_i$ for any $i$ and $\dsrk(\mathcal{C})=d$.
    Let $j$ and $\delta$ be the unique integers satisfying $d-1=\sum_{i=1}^{j-1}n_i+\delta$ and $0 \leq \delta \leq n_j-1$. Then \[\lvert \mathcal{C} \rvert \leq q^{\sum_{i=j}^tn_im_i-m_j \delta}.\]
\end{theorem}

The sum-rank metric codes whose parameters satisfy the equality in the above bound are called \textbf{Maximum Sum-Rank Distance} (or shortly \textbf{MSRD}) codes.

The interested reader is referred to \cite{byrne2021fundamental,byrne2020anticodes,gorla2023sum,moreno2021optimal} for a more details about sum-rank metric codes and their connections.

\section{Spheres, balls and sphere-packing bound in the sum-rank metric}\label{sec:balls}

In the following, we assume that the integers $n_1,\ldots,n_t,m_1,\ldots,m_t$ are such that 
\[ n_i \leq m_i, \quad i \in \{1,\ldots,t\}.\]
This is motivated by the fact that in the following sections we will deal with tha case in which all the $n_i$'s and the $m_i$'s are equal.

\begin{definition}
Let $l$ be a non-negative integer. For every vector $X \in \spacmn$, the \textbf{sphere} centered at $X$ and of radius $l$ in the sum-rank metric is defined as the following set:
$$S(X,l)=\{Y \in \spacmn \, | \, \dsrk(X,Y) = l\}.$$
\end{definition}

\begin{definition}
Let $r$ be a non-negative integer. For every vector $X \in \Pi$, the \textbf{ball} centered at $X$ and of radius $r$ in the sum-rank metric is defined as the following set:
$$B(X,r)=\{Y \in \spacmn \, | \, \dsrk(X,Y) \leq r\}.$$
\end{definition}

Note that the sphere centered at any vector $X$ in $\spacmn$ and of radius $r$ corresponds to the union of all spheres of radius between $0$ and $r$ centered in $X$:
\begin{equation} \label{c2}
	B(X,r)=\bigcup_{l=0}^r S(X,l).
\end{equation}

Let us denote $\mathbf{V}(B(X,r)) = |B(X,r)|$, and we will refer to the cardinality of a ball as its \textbf{volume}. Clearly, the size of the spheres and balls in the sum-rank metric are independent of the chosen center. Therefore, the quantity $\mathbf{V}(B(X,r))$ represents the volume of any ball with a given center and radius $r$. 

Thus, it is legitimate to consider the volume of a ball of radius $r$ in $\spacmn$ as the volume of the ball centered at the origin. To refer to this quantity, we will use the notation $\mathbf{V}_r(\spacmn)$, emphasizing that the ball lies in the space $\spacmn$.

The same considerations apply to spheres. For this reason, the spherical surface centered at the origin and of radius $l$ is given by the set:
$$
S_l = \{Y \in \spacmn \, | \, w_{sr}(Y) = l \},
$$
so that we can define its volume as
$$
\mathbf{V}(S_l) = |\{Y \in \spacmn \, | \, w_{sr}(Y) = l \}|.
$$

Recalling \eqref{c2}, we obtain:
\begin{equation} \label{c3} 
	\mathbf{V}_r(\spacmn) = \sum_{l=0}^{r} \mathbf{V}(S_l).
\end{equation}

By counting the number of the elements in $\spacmn$ having having a fixed sum-rank weight, in \cite[Section III]{byrne2021fundamental} the authors proved the following (see also \cite[Section 4]{sauerbier2025bounds}).

\begin{proposition} \label{c6}
Let $l$ and $r$ be two non-negative integers. Then
\[
\mathbf{V}(S_l) = \sum_{\substack{(k_1,\ldots,k_t) \in \mathbb{N}_0^t \\ k_1+\ldots+k_t=l}} \prod_{i=1}^{t} \qqbin{n_i}{k_i} \prod_{j=0}^{k_i-1}(q^{m_i}-q^j),
\]
and
\[
\mathbf{V}_r(\spacmn) = \sum_{l=0}^{r} \sum_{\substack{(k_1,\ldots,k_t) \in \mathbb{N}_0^t \\ k_1+\ldots+k_t=l}} \prod_{i=1}^{t} \qqbin{n_i}{k_i} \prod_{j=0}^{k_i-1}(q^{m_i}-q^j).
\]
\end{proposition}

The above formula is not very practical and difficult to compute, especially when studying covering properties; see also \cite{puchinger2022generic} for a dynamic program to compute it. 
We explicitly provide it for a very special cases in the following proposition, which is useful for the next results and is also expository for showing how this formula is intricate.

\begin{proposition}\label{prop:volume1and2}
Suppose that $ n_1 = \ldots = n_t = n $ and $ m_1 = \ldots = m_t = n $. In $\spacmn$ the following hold
\begin{equation}\label{sfera1}
\mathbf{V}_1(\spacmn)=1+ t \, 
	\frac{(q^n-1)^2}{q-1},
\end{equation}
and
\begin{equation}\label{sfera}
\mathbf{V}_2(\spacmn)=1+t\frac{(q^n-1)^2}{q-1}+\frac{t(t-1)}{2}\frac{(q^n-1)^4}{(q-1)^2}+\frac{t(q^n-1)^2(q^{n-1}-1)^2q}{(q^2-1)(q-1)}.
\end{equation}
\end{proposition}
\begin{proof}
From Proposition \ref{c6}, we derive that
$$ \mathbf{V}_1(\spacmn)= \sum_{l=0}^{1} \sum_{\substack{(k_1,\ldots,k_t) \in \mathbb{N}_0^t \\ k_1+\ldots +k_t=l}} \prod_{i=1}^{t} \qqbin{n}{k_i} \prod_{j=0}^{k_i-1} (q^n-q^j). $$
Clearly, $k_1+\ldots +k_t=0$ implies that all the $k_i$'s are zero and $k_1+\ldots +k_t=1$ implies that all but one of the $k_i$'s are zero and the remaining one is one.
Using this, one can immediately observe that $\mathbf{V}_1(\spacmn)=1+ t \, 
	\frac{(q^n-1)^2}{q-1}$.

Again, using Proposition \ref{c6}, we derive that
\[
\mathbf{V}_2(\spacmn)=\sum_{l=0}^{2}\sum_{{\substack{(k_1,\ldots,k_t) \in \mathbb{N}_0^t \\ k_1+\ldots +k_t=l}}}\prod_{i=1}^{t}\qbinom{n}{k_i}\prod_{j=0}^{k_i-1}(q^n-q^j).
\]
Clearly, we have that
\begin{equation}\label{eq:sum3}
\mathbf{V}_2(\spacmn)=1+t\frac{(q^n-1)^2}{q-1}+\sum_{\substack{(k_1,\ldots,k_t) \in \mathbb{N}_0^t \\ k_1+\ldots +k_t=2}}\prod_{i=1}^{t}\qbinom{n}{k_i}\prod_{j=0}^{k_i-1}(q^n-q^j).
\end{equation}
Note that we have that
\[
k_1+\cdots+k_t=2
\]
if and only if one of the following cases hold
\begin{enumerate}
\item[(1)] there exists $s \in \{1,\ldots,t\}$ such that $k_s=2$, $k_j=0$ for every $j\neq s$;
\item[(2)] there exist $s,u \in \{1,\ldots,t\}$ such that $k_s=k_u=1$, $k_j=0$ for every $j\neq s,u$.
\end{enumerate}
\textbf{Case (1)}\\
Suppose that there exists $s \in \{1,\ldots,t\}$ such that $k_s=2$, $k_j=0$ for every $j\neq s$.
Note that 
\[\prod_{i=1}^{t}\qbinom{n}{k_i}\prod_{j=0}^{k_i-1}(q^n-q^j)=\qbinom{n}{2}(q^n-1)(q^n-q),  \]
as $\qbinom{n}{k_j}=1$ for any $j \ne s$.\\
\textbf{Case (2)}\\
Assume that there exist $s,u \in \{1,\ldots,t\}$ such that $k_s=k_u=1$, $k_j=0$ for every $j\neq s,u$.
Thus, we have that
\[
\prod_{i=1}^{t}\qbinom{n}{k_i}\prod_{j=0}^{k_i-1}(q^n-q^j)=\left(\frac{(q^n-1)}{q-1}\times (q^n-1)\right)^2=\frac{(q^n-1)^4}{(q-1)^2}.
\]

For the Case (1) there are exactly $t$ choices for $k_s$, so that 
\begin{equation}\label{eq:sum1}
\sum_{\substack{(k_1,\ldots,k_t) \in \{0,1\}^t \\ k_1+\ldots +k_t=2}}\prod_{i=1}^{t}\qbinom{n}{k_i}\prod_{j=0}^{k_i-1}(q^n-q^j)=t\qbinom{n}{2}(q^n-1)(q^n-q)=qt\frac{(q^n-1)^2(q^{n-1}-1)^2}{(q^2-1)(q-1)}.
\end{equation}

For Case (2), there are $\binom{t}{2}=\frac{t(t-1)}{2}$ choices for $k_u$ and $k_s$, therefore we get 
\begin{equation}\label{eq:sum2}
\sum_{\substack{(k_1,\ldots,k_t) \in \{0,2\}^t \\ k_1+\ldots +k_t=2}}\prod_{i=1}^{t}\qbinom{n}{k_i}\prod_{j=0}^{k_i-1}(q^n-q^j)=\frac{t(t-1)}{2}\frac{(q^n-1)^4}{(q-1)^2}.
\end{equation}
By summing \eqref{eq:sum1} and \eqref{eq:sum2} together with \eqref{eq:sum3}, we obtain the value of $\mathbf{V}_2(\spacmn)$.
\end{proof}

Later, we will see some bounds that will allow us to prove some non-existence results for perfect codes in the sum-rank metric.
In this context the sphere-like bound reads as follows, see \cite[Theorem III.6]{byrne2021fundamental}.

\begin{theorem} \label{H} \textbf{(Sphere-packing bound)} \\
Let $\mathcal{C}$ be a sum-rank code in $\spacmn$, and let $\dsrk(\mathcal{C}) = d$ be its minimum distance. Then
\[
|\mathcal{C}| \leq \left\lfloor \frac{|\spacmn|}{\mathbf{V}_r(\spacmn)} \right\rfloor
= \left\lfloor \frac{q^{\sum_{i=1}^{t} m_i n_i}}{
\sum_{l=0}^{r} \sum_{\substack{(k_1,\ldots,k_t) \in \mathbb{N}_0^t \\ k_1 + \ldots + k_t = l}} 
\prod_{i=1}^{t} \qqbin{n_i}{k_i} \prod_{j=0}^{k_i - 1}(q^{m_i} - q^j)
} \right\rfloor,
\]
where $ r = \lfloor \frac{d - 1}{2} \rfloor $.
\end{theorem}

Codes for which equality holds in the previous bound are called \textbf{perfect codes}.
We know that such codes exist when $n_1 = \ldots = n_t = 1$, or $m_1 = \ldots = m_t = 1$, that is, in the Hamming metric. For $t = 1$, i.e., when considering the rank metric, it is known that there are no non-trivial perfect codes, as proved by Loidreau in \cite{loidreau2006properties}. As Martinez-Pe\~{n}as
 pointed out in \cite{martinezpenas2021hamming}, it seems very challenging to find examples of perfect codes that do not originate from the Hamming metric.
As far as we know, at the moment only examples of quasi-perfect codes have been found by Chen in \cite{chen2024quasi}.

Our aim is to exclude the existence of perfect codes in the sum-rank metric for several values of the parameters, by following the approach of Loidreau in \cite{loidreau2006properties}, which combines bounds on the volumes of the balls and the Singleton-like bound.

In the rank metric, the following bounds have been derived.

\begin{proposition} \label{d1}(see \cite[Proposition 1]{loidreau2006properties})
Let $\mathbf{n}=n$ and $\mathbf{m}=m$, and let $ k $ be a positive integer such that $ k \leq \min\{n, m\} $. Then:
\[
q^{(n + m - 2)k - k^2} \leq \mathbf{V}(S_k) \leq q^{(n + m + 1)k - k^2},
\]
\[
q^{(n + m - 2)k - k^2} \leq \mathbf{V}_r(\mathrm{Mat}(n,m,\fq)) \leq q^{(n + m + 1)k - k^2 + 1}.
\]
\end{proposition}

The lower bounds on the volume of spheres and balls in the rank metric given in Proposition \ref{d1} can also be extended to the sum-rank metric. In this case, we assume that $ n_1 = \ldots = n_t = n $ and $ m_1 = \ldots = m_t = m $, so that $ \spacmn $ is the direct sum of $ t $ matrix spaces of size $ n \times m $, that is,
\[
\spacmn = \mathrm{Mat}(n,m,\fq) \oplus \ldots \oplus \mathrm{Mat}(n,m,\fq).
\]

The following inequalities on the volumes of balls and spheres in the sum-rank metric can also be derived from those in \cite{sauerbier2025bounds,puchinger2022generic,ott2021bounds}, and note that the bounds we derive are even less precise. For our purposes, we need lower and upper bounds that approximate the volumes of balls and spheres as $q$-power expressions as closely as possible. Therefore, for the sake of completeness, by using a similar argument to \cite{ott2021bounds} we prove the following bounds and present them in the form we have established.

\begin{proposition} \label{d9}
Let $ k $ be a positive integer,  $ n_1 = \ldots = n_t = n $ and $ m_1 = \ldots = m_t = m $. In $\spacmn$ the following inequalities hold
\[
\mathbf{V}_k(\spacmn) \geq \mathbf{V}(S_k) \geq q^{(m + n - \frac{k}{t} - 2)k - \frac{t}{4}}.
\]
\end{proposition}
\begin{proof}
	In this case, the volume of a sphere can be written as follows 
    $$\mathbf{V}(S_k)=\sum_{\substack{(k_1,\ldots,k_t) \in \mathbb{N}_0^t \\ k_1+\ldots+k_t=k}} \prod_{i=1}^{t}\qqbin{n}{k_i} \prod_{j=0}^{k_i-1}(q^{m}-q^j),$$ 
    and therefore, by Proposition \ref{d1}, we have 
    $$ \sum_{\substack{(k_1,\ldots,k_t) \in \mathbb{N}_0^t \\ k_1+\ldots+k_t=k}} \prod_{i=1}^{t}\qqbin{n}{k_i} \prod_{j=0}^{k_i-1}(q^{m}-q^j) \geq \sum_{\substack{(k_1,\ldots,k_t) \in \mathbb{N}_0^t \\ k_1+\ldots+k_t=k}} q^{(n+m-2)\sum_{i=1}^{t}k_i-\sum_{i=1}^{t}k_i^2} \geq $$ $$ \geq  \max_{\substack{k_1+\ldots +k_t=k}} \left\{ q^{(n+m-2)\sum_{i=1}^{t}k_i-\sum_{i=1}^{t}k_i^2} \right\}=q^{(n+m-2)k} \, \max_{\substack{k_1+\ldots +k_t=k}} \left\{q^{-\sum_{i=1}^{t}k_i^2}\right\}=$$ $$= q^{(n+m-2)k} \, q^{-\min_{\substack{k_1+\ldots +k_t=k}} \left\{\sum_{i=1}^{t} k_i^2 \right\}}. $$ 
    The integer $k$ can be written as $$k=k_* t + h,$$ with $ 0 \leq h < t$; hence we can consider the tuple $(k_1,\ldots,k_t)$, with $k_{i_1}=\ldots=k_{i_h}= k_*+1$ and $k_{i_{h+1}}=\ldots=k_{i_t}=k_*$, and observe that $$k= h(k_*+1)+(t-h)k_*=hk_*+h+tk_*-hk_*=tk_*+h.$$ 
    One can check that this tuple minimizes the quantity $\sum_{i=1}^{t}k_i^2$ with the restriction that $k_1+\ldots +k_t=k$. 
    Therefore,
    \begin{equation}\label{eq:minimize}
    \min_{\substack{k_1+\ldots +k_t=k}} \left\{\sum_{i=1}^{t} k_i^2 \right\} = h (k_*+1)^2+(t-h)k_*^2=h+2hk_*+tk_*^2. 
    \end{equation}
    Since $k_*= \frac{k-h}{t}$, \eqref{eq:minimize} becomes 
    $$h + 2h \frac{(k-h)}{t} + t \frac{(k-h)^2}{t^2}=$$ $$= \frac{t^2h+2ht(k-h)+t(k-h)^2}{t^2}= h+ \frac{2hk-2h^2}{t}+ \frac{k^2+h^2-2hk}{t} = \frac{k^2-h^2}{t} +h .$$ Hence, 
    $$ \min_{\substack{k_1+\ldots +k_t=k}} \left\{\sum_{i=1}^{t} k_i^2 \right\} = \frac{k^2-h^2}{t} +h .$$
    Also, we have that 
    $$\max_{\substack{h \leq t-1}} \left\{h-\frac{h^2}{t} \right\} \leq \frac{t}{4}.$$ 
	Therefore, 
    $$q^{(n+m-2)k} \, q^{-\min_{\substack{k_1+\ldots +k_t=k}} \left\{\sum_{i=1}^{t} k_i^2 \right\}}= q^{(n+m-2)k} \, q^{-\left(\frac{k^2-h^2}{t}+h \right)} \geq$$ $$ \geq q^{(n+m-2)k} \, q^{-\frac{k^2}{t}-\frac{t}{4}}= q^{(m+n-\frac{k}{t}-2)k-\frac{t}{4}}. $$ The assertion then follows from the fact that $\mathbf{V}_k(\spacmn) \geq \mathbf{V}(S_k) $. 

\end{proof}

The integers involved in the $ t $-tuple of partitions $(k_1, \ldots, k_t)$, with $ k_1 + \ldots + k_t = k $, where $ k $ denotes the radius of the sphere considered in the sum-rank metric, are known to satisfy
\[
k_i \leq \min\{n, m\}
\]
for every $ i \in \{1, \ldots, t\} $. Setting $ n = \min\{n, m\} $, it is possible to define the set of all such ordered partitions with respect to the above constraint:
\[
\tau_{k, t, n} = \left\{ k = [k_1, \ldots, k_t] \mid \sum_{i=1}^t k_i = k, \quad k_i \leq n \right\},
\]
where the notation $ k = [k_1, \ldots, k_t] $ indicates the integers $ k_i $ involved in the partition.

The cardinality of this set corresponds to the number of ways the radius $ k $ can be decomposed into a $ t $-tuple in which each component is at most $ n $. For this, one has the bound:
\begin{equation} \label{d10}
	|\tau_{k, t, n}| \leq \binom{k + t - 1}{t - 1},
\end{equation}
see \cite{ratsaby2008estimate} for further details. These considerations, together with Proposition \ref{d1}, allow us to derive an upper bound for the volume of spheres and balls in the sum-rank metric.

\begin{proposition} \label{d12}
Let $ k $ be a positive integer,  $ n_1 = \ldots = n_t = n $ and $ m_1 = \ldots = m_t = m $. In $\spacmn$ the following inequalities hold
\[
\mathbf{V}(S_k) \leq \mathbf{V}_k(\spacmn) \leq k(k+1) \binom{k+t-1}{t-1} q^{(n + m + 1 - \frac{k}{t})k + \frac{4 - t}{4}}.
\]
\end{proposition}
\begin{proof}
	By Propositions \ref{d1} and \ref{d9} we have that 
    $$ \mathbf{V}_k(\spacmn)= \sum_{l=0}^{k} \sum_{\substack{(k_1,\ldots,k_t) \in \mathbb{N}_0^t \\ k_1+\ldots+k_t=l}} \prod_{i=1}^{t} \qqbin{n}{k_i} \prod_{j=0}^{k_i-1}(q^{m}-q^j) \leq $$ $$$$ $$ \leq (k+1) \sum_{\substack{(k_1,\ldots,k_t) \in \mathbb{N}_0^t \\ k_1+\ldots+k_t=k}} q^{(n+m+1)k-\sum_{i=1}^{t} k_i^2+1} = $$ $$$$ $$ = (k+1)  \sum_{k \in \tau_{k,t,n}} q^{(n+m+1)k-\sum_{i=1}^{t} k_i^2+1}.$$ 
    By making use of \eqref{d10} we derive that
    $$ \mathbf{V}_k(\spacmn)\leq  k(k+1) \binom{k+t-1}{t-1} q^{(n+m+1)k+1} \max\{ q^{-\sum_{i=1}^{t} k_i^2}\} = $$ $$$$ $$=k(k+1) \binom{k+t-1}{t-1} q^{(n+m+1)k+1}  q^{-\min \{\sum_{i=1}^{t} k_i^2\}} \leq $$ $$$$ $$ \leq k(k+1) \binom{k+t-1}{t-1} q^{(n+m+1)k+1 -\frac{k^2}{t}-\frac{t}{4}} = $$ $$$$ $$= k(k+1) \binom{k+t-1}{t-1} q^{(n+m+1-\frac{k}{t})k+\frac{4-t}{4}}. $$	
\end{proof}

We underline that the above bounds are not tight and improvements to them have been provided in \cite{sauerbier2025bounds,puchinger2022generic,ott2021bounds}.

\section{Perfect codes in \texorpdfstring{$\mathbb{F}_q^{n\times n}\times \mathbb{F}_q^{n\times n}$}{Fq(n×n) × Fq(n×n)}}\label{sec:t=2}

As previously stated, Loidreau demonstrated in \cite{loidreau2006properties} that when $t=1$ (i.e., when the sum-rank metric corresponds to the rank metric), non-trivial perfect codes do not exist. Consequently, the first open case in the sum-rank metric pertains to the situation where $\spacmn$ is the direct sum of two square matrix spaces. This section focuses on studying this particular case, in which we manage to rule out a large set of parameters concerning their existence.
We divide the analysis into two main cases according to the large or small values of the minimum distance.

\begin{theorem}\label{thm:nonexistence1t=2}
Let $t=2$, $ n_1 = n_2 = n $ and $ m_1 = m_2 = n $ be positive integers and $\spacmn=\mathbb{F}_q^{n\times n}\oplus\mathbb{F}_q^{n\times n}$. If $q\geq 5$, $\lfloor \frac{d-1}2\rfloor \geq 6$ then there do not exist perfect codes in $\spacmn$ with minimum distance $d$. 
\end{theorem}
\begin{proof} 
Suppose $\mathcal{C}$ is a perfect code in $\spacmn$ with minimum distance $d$. 
If $d$ is odd, then $k=\frac{d-1}{2}$ and $d=2k+1$.
The Singleton-like bound (cfr. Theorem \ref{th:SingletonboundmatrixRav}) reads as follows 
\begin{equation}\label{eq:singboundt=2}
\lvert \mathcal{C}\rvert\leq q^{2n^2-2kn}
\end{equation}
If $d$ is even, then $k=\frac{d-2}{2}$ and $d=2k+2$. The Singleton-like bound now reads as 
\begin{equation}\label{eq:singboundt=22}
\lvert \mathcal{C}\rvert\leq q^{2n^2-(2k+1)n},
\end{equation}
implying that in both the cases the inequality  $\lvert \mathcal{C}\rvert\leq q^{2n^2-2kn}$ holds and so we will use it to consider both the cases simultaneously.
Since $\C$ is perfect, we have
	\[
\lvert \mathcal{C}\rvert \mathbf{V}_k(\spacmn) =\lvert \mathbb{F}_q^{n\times n}\times \mathbb{F}_q^{n\times n}\rvert=q^{2n^2},\]
and by \eqref{eq:singboundt=2} we have
\[\lvert \mathcal{C}\rvert \mathbf{V}_k(\spacmn) \leq q^{2n^2-2kn}\mathbf{V}_k(\spacmn),
\]
that is $q^{2kn}\leq \mathbf{V}_k(\spacmn)$.
Proposition \ref{d12}, with $n=m$ and $t=2$, implies that 
\[
\mathbf{V}_k(\spacmn)\leq k(k+1)\binom{k+2-1}{2-1}q^{(2n+1-\frac{k}{2})k+\frac{4-2}{4}}=k(k+1)\binom{k+1}{1}q^{(2n+1-\frac{k}{2})k+\frac{1}{2}}
\]
from which it follows
\[
	q^{2kn}\leq k(k+1)^2q^{(2n+1-\frac{k}{2})k+\frac{1}{2}}
\]
i.e.
\begin{equation}\label{condizione su k}
q^{\frac{k^2-2k-1}{2}}\leq k(k+1)^2.
\end{equation}
By taking the base $q$ logarithm in the above inequality we get
\begin{equation}\label{condizione su k2}
\frac{k^2-2k-1}{2}\leq \lg_q(k)+2\lg_q(k+1)=\frac{\ln(k)+2\ln(k+1)}{\ln(q)}
\end{equation}
Applying the following well-known inequalities
\[
\ln(1+x)\leq x \ \ \forall x>-1,
\]
\[
\ln(x)\leq x-1 \ \ \forall x>0,
\]
we obtain
\[
\ln(k)+2\ln(k+1)\leq k-1+2k=3k-1
\]
and so \eqref{condizione su k2} becomes
\[
\frac{k^2-2k-1}{2}\leq \frac{3k-1}{\ln(q)}
\]
i.e.
\[
k^2\ln(q)-2k\ln(q)-\ln(q)-6k+2\leq0\implies k^2\ln(q)-(2\ln(q)+6)k-\ln(q)+2\leq0,
\]
that is
\begin{equation}\label{eq:f(k)}
f(k)\leq 0,
\end{equation}
where
	\[
f(k)=k^2\ln(q)-(2\ln(q)+6)k-\ln(q)+2.
\]
Since $q\geq 5$ we also have that $\ln(q)>\frac{34}{23}$ and
	\[
f(6)=36\ln(q)-12\ln(q)-36-\ln(q)+2=23\ln(q)-34>0\iff \ln(q)>\frac{34}{23},
\]
implying that $f(6)>0$.
Observe that
	\[
f'(k)=2k\ln(q)-2\ln(q)-6\geq0\iff k\geq\frac{2\ln(q)+6}{2\ln(q)}=1+\frac{3}{\ln(q)}.
\]
Since $q\geq 5$ we have $1+\frac{3}{\ln(q)}<6$, from which it follows that
\[
1+\frac{3}{\ln(q)}<6<k.
\]
Therefore, for $k\geq6$ the function $f(k)$ increases. Finally, by taking the limit as $k\to\infty$ we obtain
\[
\lim\limits_{k\to\infty}f(k)=+\infty
\]
We then conclude that
\[
f(k)>0 \ \ \forall k\geq 6
\]
Therefore, we have a contradiction to \eqref{eq:f(k)}.
\end{proof}

We now consider the case where the minimum distance is at most $12$.

\begin{theorem}\label{thm:nonexistence2t=2}
Let $t=2$, $ n_1 = n_2 = n $ and $ m_1 = m_2 = n $ be positive integers and $\spacmn=\mathbb{F}_q^{n\times n}\oplus\mathbb{F}_q^{n\times n}$. If $q> e^3$ and $\lfloor \frac{d-1}2\rfloor\leq 5$,  then there do not exist perfect codes in $\spacmn$ with minimum distance $d$.
\end{theorem}
\begin{proof}
Suppose $\mathcal{C}$ is a perfect code in $\spacmn$ with minimum distance $d$.
We can again assume that $\lvert \mathcal{C}\rvert\leq q^{2n^2-2kn}$ (from the Singleton-like bound) and $q^{\frac{k^2-2k-1}{2}}\leq k(k+1)^2$ (from \eqref{condizione su k}).
The latter inequality implies that
\[
q^{\frac{k^2-2k-1}{2}}\leq k(k+1)^2\leq 180.
\]
By considering the base $q$ logarithm in the above inequality, we obtain
\[
\frac{k^2-2k-1}{2}\leq \lg_q(180)
\]
i.e.
\[
0\geq \frac{k^2}{2}-k-\frac{1}{2}-\lg_q(180)\geq \frac{k^2}{2}-5-\frac{1}{2}-\lg_q(180)=\frac{k^2}{2}-\frac{11}{2}-\frac{\ln(180)}{\ln(q)}.
\]
This condition does not hold for $\ln(q)>3$ since
\[
(k^2-11)\ln(q)>2\ln(180) \iff \ln(q)>\frac{2\ln(180)}{k^2-11}
\]
and
\[
\frac{2\ln(180)}{k^2-11}<\frac{11}{k^2-11}<3<\ln(q),
\]
a contradiction to the assumption that $q\geq e^3$.

\end{proof}

Therefore, as a consequence, we restrict the existence of perfect codes in the sum-rank metric in the following cases.

\begin{corollary}
    Let $t=2$, $ n_1 = n_2 = n $ and $ m_1 = m_2 = n $ be positive integers and $\spacmn=\mathbb{F}_q^{n\times n}\oplus\mathbb{F}_q^{n\times n}$.
    If a perfect code with minimum distance $d$ exists, then one of the following hold:
    \begin{itemize}
        \item $5 \leq q < e^3$ and $\lfloor \frac{d-1}2\rfloor\leq 5$;
        \item $q\in \{2,3,4\}$.
    \end{itemize}
\end{corollary}

We are able to exclude the case in which $q=2$, and $q=3$ for large values of the minimum distance.

\begin{proposition}\label{prop:q=2,3nonexistencecomplete}
Let $t=2$, $q\in \{2,3\}$, $ n_1 = n_2 = n $ and $ m_1 = m_2 = n $ be positive integers and $\spacmn=\mathbb{F}_q^{n\times n}\oplus\mathbb{F}_q^{n\times n}$.
If one of the following conditions hold $\lfloor \frac{d-1}2\rfloor \geq 5$, 
\begin{itemize}
    \item $q=2$, $n\geq 4$ and $\lfloor \frac{d-1}2\rfloor \ne 2$;
    \item $q=3$ and $n\geq 3$,
\end{itemize}
then there do not exist perfect codes in $\spacmn$ with minimum distance $d$.
\end{proposition}
\begin{proof}
Suppose $\mathcal{C}$ is a perfect code in $\spacmn$ with minimum distance $d$.  \\
\textbf{Case 1}: $q=2$.\\
Arguing as in the proof of Theorem \ref{thm:nonexistence1t=2}, from \eqref{condizione su k} we have that
\[
q^{\frac{k^2-2k-1}{2}}\leq k(k+1)^2,
\]
and for $q=2$ it reads as
\[
2^{\frac{k^2-2k-1}{2}}\leq k(k+1)^2.
\] 
Simple computations show that the above inequality yields a contradiction when $k\geq 6$.
For the cases $2\leq k\leq 5$, see the Appendix \ref{app}. 
\\
\textbf{Case 2}: $q=3$.\\
Again from \eqref{condizione su k} we have that
\[
3^{\frac{k^2-2k-1}{2}}\leq k(k+1)^2,
\]
which yields a contradiction when $k\geq5$. For $2\leq k \leq 4$, see Appendix \ref{app:q=32blocks}.
\end{proof}

\section{Larger number of blocks}\label{sec:t>2}

In this section we derive some non-existence results for perfect codes in $\spacmn$ when $t>2$ and $n>1$.
Clearly, if $\lfloor \frac{d-1}2\rfloor=0$ (i.e. when $d=1$ or $2$), then we have that a perfect code with minimum distance $d$ needs to be the entire space $\spacmn$. Therefore, we start with the case where $\lfloor \frac{d-1}2\rfloor=1$ and $\lfloor \frac{d-1}2\rfloor=2$, where we are going to use the explicit value of the $\mathbf{V}_1(\spacmn)$ and $\mathbf{V}_2(\spacmn)$ from Proposition \ref{prop:volume1and2}.
Then we will move some general results, where we prove some non-existence result on perfect codes in $\spacmn$ when the required minimum distance is large enough with respect to the number of blocks and finally we prove some non-existence results by finding some congruences involving $q$, $t$ and the minimum distance.

\begin{proposition}\label{prop:d-1/2=1}
Let $t$, $ n_1 = \ldots = n_t = n $ and $ m_1 = \ldots = m_t = n $ be positive integers and $\spacmn=\underbrace{\mathbb{F}_q^{n\times n}\oplus \ldots \oplus \F_q^{n\times n}}_{t \text{ times}}$.
If $q> t+1$, there do not exist perfect codes in $\spacmn$ having minimum distance $d$ with $\lfloor \frac{d-1}2\rfloor=1$.
\end{proposition}
\begin{proof}
Suppose that there exists a perfect code $\C$ in $\spacmn$ with minimum distance $d$.
By Proposition \ref{prop:volume1and2}, we have that
\[
\mathbf{V}_1(\spacmn)=1+t\frac{(q^n-1)^2}{q-1}.
\]
Since $\C$ is a perfect code, we have
\[
\mathbf{V}_1(\spacmn)\lvert \mathcal{C}\rvert=\lvert\underbrace{\mathbb{F}_q^{n\times n}\oplus \ldots \oplus \F_q^{n\times n}}_{t \text{ times}}\rvert=q^{tn^2}.
\]
By Singleton-like bound, we have
\[
\lvert \mathcal{C}\rvert\leq q^{tn^2-n(d-1)},
\]
implying that 
\[
\lvert \mathcal{C}\rvert\leq q^{tn^2-2n}.
\]
By plugging together the above inequalities, we obtain
\[
\mathbf{V}_1(\spacmn)\lvert\mathcal{C}\rvert=q^{tn^2}\leq \mathbf{V}_1(\spacmn)q^{tn^2-2n},
\]
i.e.
\[
q^{2n}\leq \mathbf{V}_1(\spacmn)=1+t\frac{(q^n-1)^2}{q-1},
\]
from which we derive
\[
q^{2n}(q-1)\leq (q-1)+t(q^n-1)^2,
\]
that is
\[
q^{2n+1}+2tq^n \leq (t+1)q^{2n}+q+t-1.
\]
Now, since
\[
q+t-1<2tq^n
\]
we have that 
\[
q^{2n+1}+2tq^n < (t+1)q^{2n}+2tq^n,
\]
which gives a contradiction since $q>t+1$.
\end{proof}

Now, we consider the case when $\lfloor \frac{d-1}2\rfloor=2$.

\begin{proposition}\label{prop:d-1/2=2}
Let $t$, $ n_1 = \ldots = n_t = n $ and $ m_1 = \ldots = m_t = n $ be positive integers and $\spacmn=\mathbb{F}_q^{n\times n}\times \cdots \times \mathbb{F}_q^{n\times n}$.
If one of the following conditions:
\begin{itemize}
    \item $q$ is odd, neither $t\not\equiv1\pmod{q}$ nor $t\not\equiv2\pmod{q}$;
    \item $q=2$, $4 \nmid t-2$ and $4 \nmid t-1$;
\end{itemize}
then there do not exist perfect codes in $\spacmn$ with minimum distance $d \in \{5,6\}$.
\end{proposition}
\begin{proof}
Suppose there exists a perfect code $\C$ in $\spacmn$ with minimum distance $d \in \{5,6\}$.
Since $\C$ is perfect we have that
\[
\mathbf{V}_2(\spacmn)\lvert\mathcal{C}\rvert=q^{tn^2}\implies \mathbf{V}_2(\spacmn)=q^{tn^2-\dim(\mathcal{C})}.
\]
Since $\C \ne \spacmn$, we have that $q$ divides $\mathbf{V}_2(\spacmn)$, and so $\mathbf{V}_2(\spacmn)\equiv0\pmod{q}$. 
From (\ref{sfera}) we have that
\[
V_2=1+t\frac{(q^n-1)^2}{q-1}+\frac{t(t-1)}{2}\frac{(q^n-1)^4}{(q-1)^2}+\frac{t(q^n-1)^2(q^{n-1}-1)^2q}{(q^2-1)(q-1)}.
\]
Let analyze the congruence modulo $q$ of each summand:
\begin{enumerate}
\item[(1)] We have that
\[
\frac{(q^n-1)^2}{q-1}=\frac{q^n-1}{q-1}(q^n-1)=(q^{n-1}+\cdots+q+1)(q^n-1)\equiv -1\pmod{q};
\]
implying that
\[
\frac{t(q^n-1)^2}{q-1}\equiv t(-1)=-t\pmod{q}.
\]
\item[(2)] Observe that
\[
\frac{(q^n-1)^4}{(q-1)^2}=\frac{q^n-1}{q-1}\frac{q^n-1}{q-1}(q^n-1)(q^n-1)\equiv (-1)\cdot(-1)=1\pmod{q},
\]
that is
\[
\frac{t(t-1)}{2}\frac{(q^n-1)^4}{(q-1)^2}\equiv 1\cdot \frac{t(t-1)}{2}=\frac{t(t-1)}{2}\pmod{q}.
\]
\item[3] Finally, we have
\[
\frac{t(q^n-1)^2(q^{n-1}-1)^2q}{(q^2-1)(q-1)}\equiv0 \pmod{q},
\]
as it is a multiple of $q$ since $q^2-1$ divides either $q^{n-1}-1$ or $q^n-1$ and $q-1$ divides both of them.
\end{enumerate}
Combining the previous items, we have that
\[
\mathbf{V}_2(\spacmn)\equiv0\pmod{q}\iff 1-t+\frac{t(t-1)}{2}\equiv0\pmod{q},
\]
i.e.
\begin{equation}\label{eq:condqcongru}
(t-1)\frac{t-1}{2}\equiv0\pmod q.
\end{equation}
If $q$ is odd then $2$ is invertible in $\mathbb{Z}_q$ and we obtain
\[
(t-1)(t-2)\equiv0\pmod q.
\]
Since $q$ is a power of a prime and $t-1$ and $t-2$ are coprime, we have that either $t\equiv1\pmod{q}$ or $t\equiv2\pmod{q}$.\\
If $q=2$, then \eqref{eq:condqcongru} reads as follows
\[
1-t+\frac{t(t-1)}{2}\equiv0\pmod{2}.
\]
Suppose that $t=2r$, for some $r \in \NN$. Then we have that
\[
1-2r+r(2r-1)\equiv0\pmod{2},
\]
i.e. $r\equiv1\pmod{2}$, that is $r$ is odd and $4\mid t-2$.
Similarly, if $t=2r+1$, for some $r \in \NN$, then
\[
r\equiv0\pmod{2}\implies 4|t-1.
\]
\end{proof}

\begin{remark}
In particular, if $q$ is odd then there does not exist perfect code in $\spacmn$ with minimum distance $d \in \{5,6\}$ when $t\in\{3,\cdots,q\}$.
\end{remark}

Another condition that needs to be satisfied by perfect codes in $\spacmn$ with minimum distance $d \in \{5,6\}$ regards their dimension.

\begin{proposition}\label{prop:n|dim}
Let $t$, $ n_1 = \ldots = n_t = n $ and $ m_1 = \ldots = m_t = n $ be positive integers and $\spacmn=\underbrace{\mathbb{F}_q^{n\times n}\oplus \ldots \oplus \F_q^{n\times n}}_{t \text{ times}}$.
If there exists a perfect code $\C$ in $\spacmn$ with minimum distance $d \in \{5,6\}$, then $n$ divides $\dim(\mathcal{C})$.
\end{proposition}
\begin{proof}
Since $\mathcal{C}$ is perfect and $\lfloor\frac{d-1}{2}\rfloor=2$, we have that
\[
\mathbf{V}_2(\spacmn)\lvert\mathcal{C}\rvert=q^{tn^2},
\]
and since $\lvert\mathcal{C}\rvert=q^{dim(\mathcal{C})}$, we have
\begin{equation}\label{eq:V2=qx}
\mathbf{V}_2(\spacmn)=q^{tn^2-\dim(\mathcal{C})}.
\end{equation}
From (\ref{sfera}), we know that
\[
\mathbf{V}_2(\spacmn)=1+t\frac{(q^n-1)^2}{q-1}+\frac{t(t-1)}{2}\frac{(q^n-1)^4}{(q-1)^2}+\frac{t(q^n-1)^2(q^{n-1}-1)^2q}{(q^2-1)(q-1)},
\]
which can be rewritten as follows
\[
\begin{aligned}
\mathbf{V}_2(\spacmn)&=1+(q^n-1)\left[\frac{q^n-1}{q-1}+\frac{t(t-1)}{2}\cdot\frac{(q^n-1)^3}{(q-1)^2}+tq\cdot\frac{(q^n-1)(q^{n-1}-1)^2}{(q^2-1)(q-1)}\right]=\\&=1+(q^n-1)\delta.
\end{aligned}
\]
where
\[
\delta=\frac{q^n-1}{q-1}+\frac{t(t-1)}{2}\cdot\frac{(q^n-1)^3}{(q-1)^2}+tq\cdot\frac{(q^n-1)(q^{n-1}-1)^2}{(q^2-1)(q-1)}.
\]
Let us observe that $\delta \in \NN$. Indeed,
\[
\frac{q^n-1}{q-1}=q^{n-1}+q^{n-2}+\cdots+1
\]
is clearly a positive integer, 
\[
\frac{t(t-1)}{2}\frac{(q^n-1)^3}{(q-1)^2}=\binom{t}{2}\frac{q^n-1}{q-1}\frac{q^n-1}{q-1}(q^n-1)
\]
is again clearly a positive integer, and finally
\[
tq\frac{(q^n-1)(q^{n-1}-1)^2}{(q^2-1)(q-1)}=t\qbinom{n}{2}(q^{n-1}-1)
\]
is a positive integer since $\qbinom{n}{2} \in \NN$ as it counts the number of subspaces of dimension two in an $\fq$-vector space of dimension $n$. Therefore, $\delta \in \NN$.
We can rewrite \eqref{eq:V2=qx} as follows
\[
1+(q^n-1)\delta=q^{tn^2-\dim(\mathcal{C})}\iff (q^n-1)\delta=q^{tn^2-\dim(\mathcal{C})}-1,
\]
which implies that $q^n-1$ divides $q^{tn^2-\dim(\mathcal{C})}$. This happens if and only if $n$ divides $tn^2-\dim(\mathcal{C})$, which holds if and only if $n \mid \dim(\mathcal{C})$.
\end{proof}

We now provide some general results for large values of the minimum distance. We divide the cases $d$ odd and $d$ even as the conditions on the parameters are slightly different.

\begin{theorem}\label{thm:dlarg1}
Let $t$, $ n_1 = \ldots = n_t = n $ and $ m_1 = \ldots = m_t = n $ be positive integers and $\spacmn=\underbrace{\mathbb{F}_q^{n\times n}\oplus \ldots \oplus \F_q^{n\times n}}_{t \text{ times}}$.
If 
\begin{itemize}
    \item $t>1$;
    \item $q>e^4$;
    \item $d$ is odd;
    \item $\lfloor \frac{d-1}2\rfloor\geq 3t$,
\end{itemize}
then there do not exist perfect codes in $\spacmn$ with minimum distance $d$.
\end{theorem}
\begin{proof}
Suppose that there exists a perfect code in $\spacmn$ with minimum distance $d$.
Since $d$ is odd, we can write $d=2k+1$ for some $k \in \NN$.
From Theorem \ref{th:SingletonboundmatrixRav}, we have
\[
\lvert \mathcal{C}\rvert\leq q^{tn^2-2kn}.
\]
Since $\C$ is perfect we also have
\[
\lvert \mathcal{C}\rvert \mathbf{V}_k(\spacmn)=\lvert\mathbb{F}_q^{n\times n}\times\cdots\times\mathbb{F}_q^{n\times n}\rvert=q^{tn^2}.
\]
Combining the two inequalities together, we also have that
\[
q^{tn^2}=\lvert\mathcal{C}\rvert \mathbf{V}_k(\spacmn)\leq q^{tn^2-2kn}\mathbf{V}_k(\spacmn),
\]
that is
\[
q^{2kn}\leq \mathbf{V}_k(\spacmn).
\]
By Proposition \ref{d12} we derive that
\[
\mathbf{V}_k(\spacmn)\leq k(k+1)\binom{k+t-1}{t-1}q^{(n+m+1-\frac{k}{t})k+\frac{4-t}{4}},
\]
so that
\[
q^{2kn}\leq k(k+1)\binom{k+t-1}{t-1}q^{(2n+1-\frac{k}{t})k+\frac{4-t}{4}},
\]
implying that
\[
q^{-k+\frac{k^2}{t}-(\frac{4-t}{4})}\leq k(k+1)\binom{k+t-1}{t-1}.
\]
Using that
\[
\binom{k+t-1}{t-1}\leq 2^{k+t-1},
\]
we derive 
\[
q^{-k+\frac{k^2}{t}-(\frac{4-t}{4})}\leq k(k+1)2^{k+t-1}.
\]
Passing to the logarithm with base $q$ and changing the basis of the logarithm we have
\[
-k+\frac{k^2}{t}-\left(\frac{4-t}{4}\right)\leq \frac{lg(k)+lg(k+1)+(k+t-1)lg2}{lg(q)}<\frac{lg(k)+lg(k+1)+k+t-1}{lg(q)},
\]
where the last inequality follows from the fact that $lg(2)=0.693<1$.
Using again the well-known inequalities
\[
lg(1+x)\leq x \ \ \forall x>-1,
\]
and
\[
lg(x)\leq x-1 \ \ \forall x>0,
\]
we obtain again
\[
\frac{\ln(k)+\ln(k+1)+k+t-1}{\ln(q)}\leq \frac{k-1+k+k+t-1}{\ln(q)}=\frac{3k+t-2}{\ln(q)}.
\]
Observing that
\[
lg(q)>4\implies \frac{1}{lg(q)}<\frac{1}{4},
\]
we obtain that
\[
\frac{3k+t-2}{lg(q)}\leq\frac{3k+t-2}{4}.
\]
Therefore,
\[
-k+\frac{k^2}{t}-\left(\frac{4-t}{4}\right)-\frac{3k+t-2}{4}<0,
\]
i.e.
\begin{equation}\label{eq:g<0}
-4kt+4k^2-t(4-t)-t(3k+t-2)<0\implies g(k)=4k^2-7kt-2t<0.
\end{equation}
Observe that the derivative of the function $g(k)$ is $g'(k)=8k-7t$ and 
\[
g'(k)\geq0\iff k\geq \frac{7}{8}t.
\]
Moreover,
\[
g(3t)=4(9t^2)-7(3t)t-2t=15t^2-2t>0,
\]
and
\[
\frac{7}{8}t<3t<k\iff \frac{7}{8}<3.
\]
Therefore $g(k)$ is positive in $3t$, it increases when $k>3t$ and
\[
\lim\limits_{k\to\infty}g(k)=\infty.
\]
This contradicts \eqref{eq:g<0}.
\end{proof}

With a very similar argument one can show the following.

\begin{theorem}\label{thm:dlarg2}
Let $t$, $ n_1 = \ldots = n_t = n $ and $ m_1 = \ldots = m_t = n $ be positive integers and $\spacmn=\underbrace{\mathbb{F}_q^{n\times n}\oplus \ldots \oplus \F_q^{n\times n}}_{t \text{ times}}$.
If 
\begin{itemize}
    \item $t>1$;
    \item $n>t$;
    \item $q>e^4$;
    \item $d$ is even;
    \item $\lfloor \frac{d-1}2\rfloor\geq \frac{7}4t$,
\end{itemize}
then there do not exist perfect codes in $\spacmn$ with minimum distance $d$.
\end{theorem}

We are going to give some more conditions on the parameters in order to exclude the existence of perfect codes in $\spacmn$.
To do so, we need the following auxiliary lemma on the volume of balls in the sum-rank metric.

\begin{lemma}
Let $t$, $ n_1 = \ldots = n_t = n $ and $ m_1 = \ldots = m_t = n $ be positive integers and $\spacmn=\underbrace{\mathbb{F}_q^{n\times n}\oplus \ldots \oplus \F_q^{n\times n}}_{t \text{ times}}$.
For every $k$ we have that
\[
\mathbf{V}_k(\spacmn)\equiv1+\sum_{i=1}^{k}(-1)^i\binom{t}{i}\pmod{q}.
\]
\end{lemma}
\begin{proof}
From \eqref{c3}, we have that
\[
\mathbf{V}_k(\spacmn)=\sum_{i=0}^{k}\mathbf{V}(S_i).
\]
Clearly, $\mathbf{V}(S_0)=1$ and so let us consider $\mathbf{V}(S_r)$ with $r\geq 1$.
By Proposition \ref{c6},
\[
\mathbf{V}(S_r)=\sum_{k_1+\cdots+k_t=r}\prod_{i=1}^{t}\qbinom{n}{k}\prod_{j=0}^{k_i-1}(q^n-q^j).
\]
If at least one of the $k_i$'s is greater than $1$, then the product
\[
\prod_{j=0}^{k_i-1}(q^n-q^j)
\]
is a multiple of $q$. 
Therefore, it is enough to consider the case in which all $k_i$'s are less than or equal to one.
Taking into account that we also need that
\[
k_1+\cdots k_t=r,
\]
the only possibility is that $r$ of the $k_i$'s are equal to one and the remaining ones are equal to zero. 
Hence, we have
\[
\prod_{i=1}^{t}\qbinom{n}{1}\prod_{j=0}^{0}(q^n-q^j)=\left(\frac{(q^n-1)(q^n-1)}{q-1}\right)^r=\left(\frac{(q^n-1)^2}{q-1}\right)^r
\]
as only the factors that correspond to one contribute to the product. So, we have
\[
\left(\frac{(q^n-1)^2}{q-1}\right)^r\equiv(-1)^r\pmod{q}.
\]
The number of these summands in $\mathbf{V}(S_r)$ is $\binom{t}{r}$ and so
 \[
 \mathbf{V}(S_r)\equiv (-1)^r\binom{t}{r}\pmod{q}.
 \]
In conclusion, we have 
\[
 \mathbf{V}_k(\spacmn)=\sum_{i=0}^{k}S_i\equiv1+\sum_{i=1}^{k}(-1)^i\binom{t}{i}.
 \]
\end{proof}

As a consequence, we can derive the following necessary conditions that a perfect code in $\spacmn$ needs to satisfy.

\begin{theorem}\label{thm:conddivis1}
Let $t$, $ n_1 = \ldots = n_t = n $ and $ m_1 = \ldots = m_t = n $ be positive integers and $\spacmn=\underbrace{\mathbb{F}_q^{n\times n}\oplus \ldots \oplus \F_q^{n\times n}}_{t \text{ times}}$.
Let $\mathcal{C}$ be a perfect code in $\spacmn$ with minimum distance $d$. If $\gcd(q,k!)=1$, where $k=\lfloor\frac{d-1}{2}\rfloor$, then $\gcd(t,q)=1$.
\end{theorem}
\begin{proof}
Since $\mathcal{C}$ is perfect, then
\[
\mathbf{V}_k(\spacmn)=q^{tn^2-\dim(C)}
\]
and $q$ divides $\mathbf{V}_k(\spacmn)$. 
By the above proposition, we have
\[
\mathbf{V}_k(\spacmn)\equiv 1+\sum_{i=1}^{k}(-1)^i\binom{t}{i}\equiv0\pmod{q},
\]
so that
\[
1-t+\frac{t(t-1)}{2}-\frac{t(t-1)(t-2)}{6}+\cdots+(-1)^k\frac{t(t-1)(t-2)\cdots(t-(k-1))}{k!}\equiv0 \pmod q,
\]
from which we derive
\[
1-t\frac{a}{k!}\equiv0 \pmod q,
\]
for some $a \in \NN$. Since $\gcd(k!,q)=1$, $k!$ is invertible in $\mathbb{Z}_q$ and let $b$ be its inverse. 
So, we have that
\[
1-t(ab)\equiv0\pmod{q}\implies t(ab)\equiv1\pmod q.
\]
Let $c=ab$, then
\[
tc\equiv1\pmod{q}
\]
therefore $t$ is invertible in $\mathbb{Z}_q$, which is possible if and only if $\gcd(t,q)=1$.
\end{proof}

As a consequence we derive the following non-existence results for perfect codes in $\spacmn$.

\begin{corollary}
Let $t$, $k$ $ n_1 = \ldots = n_t = n $ and $ m_1 = \ldots = m_t = n $ be positive integers and $\spacmn=\underbrace{\mathbb{F}_q^{n\times n}\oplus \ldots \oplus \F_q^{n\times n}}_{t \text{ times}}$.
If $\gcd(q,k!)=1$ and $\gcd(t,q)\ne1$, then there do not exist perfect codes in $\spacmn$ with minimum distance $d$, where $k=\lfloor\frac{d-1}{2}\rfloor$.
\end{corollary}

Using the previous lemma, we are also able to derive further non-existence results for perfect codes in $\spacmn$.

\begin{proposition}\label{prop:divisvcond2}
Let $t$, $k$ $ n_1 = \ldots = n_t = n $ and $ m_1 = \ldots = m_t = n $ be positive integers and $\spacmn=\underbrace{\mathbb{F}_q^{n\times n}\oplus \ldots \oplus \F_q^{n\times n}}_{t \text{ times}}$.
If $\gcd(q,k!)=1$ and $(t-1)\cdots(t-k)\not\equiv 0\pmod q$, then there do not exist perfect codes in $\spacmn$ with minimum distance $d$, where $k=\lfloor\frac{d-1}{2}\rfloor$.
\end{proposition}
\begin{proof}
Since $\mathcal{C}$ is perfect, arguing as in the proof of the previous result we obtain that
\begin{equation}\label{eq:condVkcongr1,...,k}
\mathbf{V}_k(\spacmn)\equiv 1+\sum_{i=1}^{k}(-1)^i\binom{t}{i}\equiv0\pmod{q}.
\end{equation}
Consider the following polynomial, 
\[
p(t)=1+\sum_{i=1}^{k}(-1)^i\binom{t}{i}=1-t+\frac{t(t-1)}{2}-\frac{t(t-1)(t-2)}{6}+\cdots+(-1)^k\frac{t(t-1)\cdots(t-(k-1))}{k!}.
\]
One can check that $p(t)=\frac{(-1)^k}{k!}(t-1)\cdots(t-k)
$, as the degree of $p$ is $k$ and has $1,\ldots,k$ as roots. Let $1\leq l\leq k$. This implies that \eqref{eq:condVkcongr1,...,k} reads
\[
\mathbf{V}_k(\spacmn)\equiv \frac{(-1)^k}{k!}(t-1)\cdots(t-k)\equiv0\pmod{q},
\]
from which we derive the assertion.
\end{proof}

\section{Computational results}\label{sec:comput}

Let $q=p^\alpha$, for some prime $p$ and $\alpha \in \NN$.
Assume also that $t$, $ n_1 = \ldots = n_t = n $ and $ m_1 = \ldots = m_t = n $ are positive integers and $\spacmn=\underbrace{\mathbb{F}_q^{n\times n}\oplus \ldots \oplus \F_q^{n\times n}}_{t \text{ times}}$.
As we observed in several proofs of the previous section, see e.g. \eqref{eq:condVkcongr1,...,k}, if there exists a perfect code $\C$ in $\spacmn$ with minimum distance $d$ and let $k=\lfloor \frac{d-1}2\rfloor$, we need that
\[
\mathbf{V}_k(\spacmn)\equiv1+\sum_{i=1}^{k}(-1)^i\binom{t}{i}\pmod{q}.
\]
In particular, it follows that if $\C\ne \spacmn$, then $\mathbf{V}_k(\spacmn)\equiv 0 \pmod p$.
This observation leads to the following general non-existence criterion.

\begin{proposition}
    Let $q=p^\alpha$, for some prime $p$ and $\alpha \in \NN$.
    Let $t$, $k$ $ n_1 = \ldots = n_t = n $ and $ m_1 = \ldots = m_t = n $ be positive integers and $\spacmn=\underbrace{\mathbb{F}_q^{n\times n}\oplus \ldots \oplus \F_q^{n\times n}}_{t \text{ times}}$.
    If $\mathbf{V}_k(\spacmn)\not\equiv 0 \pmod{p}$, then there does not exist a perfect code in $\spacmn$ with minimum distance $d$ such that $k=\lfloor \frac{d-1}2\rfloor$.
\end{proposition}

In particular, if $\mathbf{V}_k(\spacmn)\not\equiv 0 \pmod{p}$, perfect codes in $\spacmn$ cannot exist in any extension $\fq$ over $\F_p$. Therefore, equivalently, for fixed $t$ and $k$, if $\mathbf{V}_k(\spacmn)\not\equiv 0 \pmod p$, then no perfect code exists in $\spacmn$ for any extension field $\mathbb{F}_q$ with $q = p^\alpha$, independently of the value of $\alpha$.

\begin{theorem}
    For all $\alpha \in \mathbb{N}$, there do not exist perfect codes in $\spacmn$ with minimum distance $d$ and $q=2^\alpha$ for the following parameters:
    \begin{table}[ht] 
    \centering
    \label{table:p=2}
    \begin{tabular}{|c|p{0.8\textwidth}|}
	\hline
	$k=\lfloor \frac{d-1}2\rfloor$ & $t$ \\
	\hline
	3 & 4, 5, 6, 7, 12, 14, 16, 17, 18 \\
	\hline
	4 & 6, 8, 9, 10, 11, 12, 13, 15, 16, 17, 19 \\
	\hline
	5 & 8, 9, 10, 11, 12, 15, 16, 17, 19, 20 \\
	\hline
	6 & 8, 9, 10, 11, 12, 16, 17, 19 \\
	\hline
	7 & 8, 9, 10, 11, 12, 15, 16, 17, 18, 20, 22 \\
	\hline
	8 & 10, 12, 13, 14, 15, 18, 21, 23 \\
	\hline
	9 & 12, 14, 18, 20, 22, 24 \\
	\hline
	10 & 12, 14, 16, 17, 18, 21, 23 \\
	\hline
\end{tabular}
    \caption{Non-existence results for $p=2$.}
\end{table}
\end{theorem}

We have used the same approach for the characteristic $3,5$ and $7$, obtaining the following.

\begin{theorem}
    For all $\alpha \in \mathbb{N}$, there do not exist perfect codes in $\spacmn$ with minimum distance $d$ and 
    \begin{itemize}
         \item $q=3^\alpha$ for the parameters in Table \ref{table:p=3};
         \item $q=5^\alpha$ for the parameters in Table \ref{table:p=5};
         \item $q=7^\alpha$ for the parameters in Table \ref{table:p=7}.
    \end{itemize}
\end{theorem}

\begin{table}[ht]
    \centering
\begin{tabular}{|c|p{0.8\textwidth}|}
	\hline
	$k=\lfloor \frac{d-1}2\rfloor$ & $t$ \\
	\hline
	3 & 4, 5, 7, 8, 13, 14, 15, 17, 18 \\
	\hline
	4 & 6, 7, 8, 9, 10, 11, 12, 14, 15, 16, 18 \\
	\hline
	5 & 6, 7, 8, 13, 14, 15, 16, 18, 19 \\
	\hline
	6 & 7, 8, 13, 14, 16, 18, 19, 20 \\
	\hline
	7 & 8, 13, 14, 17, 18, 20, 21, 22 \\
	\hline
	8 & 9, 10, 11, 12, 14, 15, 16, 17, 18, 19, 20, 21, 22 \\
	\hline
	9 & 10, 11, 13, 14, 17, 18, 20, 21, 22 \\
	\hline
	10 & 11, 14, 15, 17, 18, 20, 21, 23, 24 \\
	\hline
\end{tabular}
    \caption{Non-existence results for $p=3$.}
    \label{table:p=3}
\end{table}

\begin{table}[ht] 
    \centering
\begin{tabular}{|c|p{0.8\textwidth}|}
	\hline
	$k=\lfloor \frac{d-1}2\rfloor$ & $t$ \\
	\hline
	3 & 4, 9, 13, 15, 16 \\
	\hline
	4 & 5, 6, 7, 8, 9, 10, 11, 12, 13, 14, 15, 16, 17, 18, 19 \\
	\hline
	5 & 6, 7, 9, 10, 11, 12, 13, 14, 15, 16, 17, 18, 19, 20 \\
	\hline
	6 & 7, 8, 9, 10, 11, 13, 14, 17, 18, 19, 21 \\
	\hline
	7 & 8, 10, 11, 12, 13, 15, 16, 17, 18, 19, 20, 21, 22 \\
	\hline
	8 & 10, 11, 12, 13, 14, 17, 18, 19, 20, 21, 23 \\
	\hline
	9 & 10, 11, 12, 13, 14, 15, 16, 17, 18, 19, 20, 21, 23, 24 \\
	\hline
	10 & 11, 12, 15, 16, 17, 18, 19, 20, 21, 22, 23, 25 \\
	\hline
\end{tabular}
    \caption{Non-existence results for $p=5$.}
        \label{table:p=5}
\end{table}

\begin{table}[ht] 
    \centering
\begin{tabular}{|c|p{0.8\textwidth}|}
	\hline
	$k=\lfloor \frac{d-1}2\rfloor$ & $t$ \\
	\hline
	3 & 4, 5, 6, 11, 12, 13, 14, 15, 17, 18 \\
	\hline
	4 & 5, 6, 7, 8, 9, 10, 11, 12, 13, 15, 16, 17, 19 \\
	\hline
	5 & 7, 8, 9, 10, 11, 12, 13, 14, 15, 16, 17, 19, 20 \\
	\hline
	6 & 7, 8, 9, 10, 11, 12, 13, 15, 16, 17, 18, 19, 20 \\
	\hline
	7 & 8, 9, 10, 11, 12, 13, 14, 15, 17, 18, 19, 20, 21, 22 \\
	\hline
	8 & 9, 10, 11, 12, 13, 14, 16, 17, 18, 19, 20, 21, 22, 23 \\
	\hline
	9 & 10, 11, 12, 13, 15, 16, 17, 18, 19, 20, 21, 22, 23, 24 \\
	\hline
	10 & 11, 12, 13, 14, 17, 18, 19, 20, 21, 22, 23, 24, 25 \\
	\hline
\end{tabular}
    \caption{Non-existence results for $p=7$.}
    \label{table:p=7}
\end{table}

\section{Conclusions}

In this work, we studied the existence of perfect codes in the sum-rank metric. Building on the classical sphere-packing framework and combining it with the Singleton-like bound, we derived several non-existence results for different parameter regimes. In particular, we extended Loidreau’s arguments from the rank metric to the sum-rank setting, showing that large families of parameters cannot admit perfect codes. For the case of two square blocks, we excluded existence both for sufficiently large and sufficiently small minimum distances, depending on the field size. We then generalized our analysis to a larger number of blocks, where explicit computations of the volumes of small-radius balls allowed us to rule out perfect codes for certain congruence conditions and divisibility constraints. Finally, computational results provided further evidence of the scarcity of perfect codes in this metric, confirming the absence of solutions for a wide range of values of $t$, $k$ and $q$.

Overall, our results strongly suggest that perfect codes in the sum-rank metric are extremely rare, and that the only natural families arise from the Hamming setting. This aligns with the intuition that the rigid combinatorial structure required by perfect codes is incompatible with the complexity of the sum-rank geometry. An interesting direction for future research is to investigate the existence and characterization of quasi-perfect codes in this metric, following the recent constructions of Chen, and to explore whether weaker notions of covering optimality might yield richer families of codes with practical applications in network coding and distributed storage.

\section{The case \texorpdfstring{$\mathbb{F}_2^{n\times n}\times \mathbb{F}_2^{n\times n}$}{F2(n×n) × F2(n×n)}}\label{app}

In this appendix, we present computations to completely exclude the existence of perfect codes in $\mathbb{F}_2^{n \times n} \oplus \mathbb{F}_2^{n \times n}$. We begin by explicitly calculating the volumes of spheres with radii $3$, $4$, and $5$ using Proposition \ref{c6}.

\begin{proposition}
Let $q=2$, $t=2$, $n_1=n_2=n$, $m_1=m_2=n$ be integers and Mat(\textbf{n},\textbf{m},$\mathbb{F}_q)=\mathbb{F}_2^{n\times n}\oplus \mathbb{F}_2^{n\times n}$. Then
\begin{equation}\label{S_3}
\textbf{V}(S_3)=2^4\cdot\frac{(2^n-1)^2(2^{n-1}-1)^2(2^{n-2}-1)^2}{(2^3-1)(2^2-1)}+2^2\cdot\frac{(2^n-1)^2(2^n-1)^2(2^{n-1}-1)^2}{2^2-1}
\end{equation}
\begin{equation}\label{S_4}
	\begin{aligned}
\textbf{V}(S_4)&=\frac{128}{315}(2^n-1)^2(2^{n-1}-1)^2(2^{n-2}-1)^2(2^{n-3}-1)^2+\\&+\frac{16}{21}(2^n-1)^2(2^n-1)^2(2^{n-1}-1)^2(2^{n-2}-1)^2+\frac{16}{3}(2^n-1)^4(2^{n-1}-1)^4
\end{aligned}
\end{equation}
\begin{equation}\label{S_5}
\begin{aligned}
\textbf{V}(S_5)&=\frac{2048}{9765}(2^n-1)^2(2^{n-1}-1)^2(2^{n-2}-1)^2(2^{n-3}-1)^2(2^{n-4}-1)^2+\\&+\frac{128}{315}(2^n-1)^2(2^n-1)^2(2^{n-1}-1)^2(2^{n-2}-1)^2(2^{n-3}-1)^2+\\&+\frac{32}{63}(2^n-1)^2(2^{n-1}-1)^2(2^n-1)^2(2^{n-1}-1)^2(2^{n-2}-1)^2
\end{aligned}
\end{equation}
\end{proposition}

We use the above proposition and the Singleton-like bound for completing the proof of Proposition \ref{prop:q=2,3nonexistencecomplete} for the case $q=2$.

\begin{theorem}
Let $q=2$, $t=2$, $n_1=n_2=n$, $m_1=m_2=n$ be integers and Mat(\textbf{n},\textbf{m},$\mathbb{F}_q$)=$\mathbb{F}_2^{n\times n}\oplus\mathbb{F}_2^{n\times n}$. If $n \geq 4$ and $\lfloor\frac{d-1}{2}\rfloor\in\{3,4,5\}$, then there do not exist perfect codes in Mat(\textbf{n},\textbf{m},$\mathbb{F}_q$) with minimum distance $d$.
\end{theorem}
\begin{proof}
Suppose $\mathcal{C}$ is a perfect code in Mat(\textbf{n},\textbf{m},$\mathbb{F}_q$) with minimum distance $d$. Denote by $k=\lfloor \frac{d-1}2\rfloor$, then
\[
\mathbf{V}_k(\text{Mat}(\textbf{n},\textbf{m},\mathbb{F}_q))\lvert\mathcal{C}\rvert=q^{2n^2}\leq q^{2n^2-2nk}
\]
i.e.
\begin{equation}\label{condizione_caso_generale}
q^{2kn}\leq \mathbf{V}_k(\text{Mat}(\textbf{n},\textbf{m},\mathbb{F}_q)).
\end{equation}
We split the proof according to the values of $k \in\{1,3,4,5\}$.\\

\textbf{Case 1:} $k=1$.\\
Using (5) we have that
\[
\mathbf{V}_1(\text{Mat}(\textbf{n},\textbf{m},\mathbb{F}_q))=1+2\cdot(2^n-1)^2.
\]
Since the code is perfect then 
\[ |\C| \mathbf{V}_1(\text{Mat}(\textbf{n},\textbf{m},\mathbb{F}_q))=2^{2n^2},  \]
i.e. 
\[ |\C|=\frac{2^{2n^2}}{ \mathbf{V}_1(\text{Mat}(\textbf{n},\textbf{m},\mathbb{F}_q))}=\frac{2^{2n^2}}{1+2\cdot(2^n-1)^2},  \]
which turns out to be not an integer, a contradiction.

\textbf{Case 2:} $k=3$.\\
From (5) we have that
\[
\mathbf{V}_2(\text{Mat}(\textbf{n},\textbf{m},\mathbb{F}_q))=1+2\cdot(2^n-1)^2+2^2\cdot\frac{(2^n-1)^2(2^{n-1}-1)^2}{3}+(2^n-1)^4.
\]
From
\begin{equation}\label{disuguaglianza_a}
(2^a-1)^2<(2^a)^2=2^{2a}\ \ \ \forall a>-1,
\end{equation}
it follows that
\[
\begin{aligned}
\mathbf{V}_2(\text{Mat}(\textbf{n},\textbf{m},\mathbb{F}_q))&\leq 1+2\cdot 2^{2n}+\frac{4}{3}\cdot2^{2n}\cdot2^{2n-2}+2^4n\leq 1+2^{2n+1}+2\cdot2^{4n-2}+2^4n\\&=1+2^{2n+1}+2^{4n-1}+2^{4n}\leq1+3\cdot2^{4n}.
\end{aligned}
\]
Using this inequality, together with (\ref{condizione_caso_generale}) and (\ref{S_3}), we can derive the following conditions
\[
2^{6n}\leq\mathbf{V}_3(\text{Mat}(\textbf{n},\textbf{m},\mathbb{F}_q))=\mathbf{V}_2(\text{Mat}(\textbf{n},\textbf{m},\mathbb{F}_q))+\mathbf{V}(S_3)\leq 1+3\cdot2^{4n}+\textbf{V}(S_3),
\]
and so
\[
	\begin{aligned}
	2^{6n}&\leq 1+3\cdot2^{4n}+\textbf{V}(S_3)=1+3\cdot2^{4n}+2^4\cdot\frac{(2^n-1)^2(2^{n-1}-1)^2(2^{n-2}-1)^2}{(2^3-1)(2^2-1)}+\\&+2^2\cdot\frac{(2^n-1)^2(2^n-1)^2(2^{n-1}-1)^2}{2^2-1}\leq 1+3\cdot2^{4n}+2^{6n-6}+\frac{4}{3}2^{6n-2}\leq\\&\leq1+3\cdot2^{4n}+2^{6n-2}\left(1+\frac{4}{3}\right)=1+3\cdot2^{4n}+\frac{7}{3}2^{6n-2},
\end{aligned}
\]
i.e.
\[
2^{6n}-1\leq3\cdot2^{4n}+\frac{7}{3}2^{6n-2},
\]
which yields a contradiction when $n\geq2$.\\

\textbf{Case 3:} $k=4$.\\
Equation (\ref{condizione_caso_generale}) becomes
\[
2^{8n}\leq\textbf{V}_4=\textbf{V}_3+\textbf{V}(S_4).
\]
Using equation (\ref{S_4}), we obtain the following condition
\[
\begin{aligned}
2^{8n}\leq \mathbf{V}_3(\text{Mat}(\textbf{n},\textbf{m},\mathbb{F}_q))+\mathbf{V}(S_4)\leq1+ 3\cdot2^{4n}+\frac{7}{3}2^{6n-2}+\mathbf{V}(S_4).
\end{aligned}
\]
Combining (\ref{S_4}) and (\ref{disuguaglianza_a}), we obtain an upper bound on $\mathbf{V}(S_4)$:
\[
\begin{aligned}
\mathbf{V}(S_4)&=\frac{128}{315}(2^n-1)^2(2^{n-1}-1)^2(2^{n-2}-1)^2(2^{n-3}-1)^2+\\&+\frac{16}{21}(2^n-1)^2(2^n-1)^2(2^{n-1}-1)^2(2^{n-2}-1)^2+\frac{16}{3}(2^n-1)^4(2^{n-1}-1)^4\leq\\&\leq2^{8n-12}+2^{8n-6}+\frac{16}{3}2^{8n-4}\leq2^{8n-4}\left(1+1+\frac{16}{3}\right)=\frac{22}{3}2^{8n-4}.
\end{aligned}
\]
Therefore,
\[
2^{8n}-1\leq3\cdot2^{4n}+\frac{7}{3}2^{6n-2}+\textbf{V}(S_4)\leq 3\cdot2^{4n}+\frac{7}{3}2^{6n-2}+\frac{22}{3}2^{8n-4}
\]
This condition does not hold for $n\geq2$.\\

\textbf{Case 4:} $k=5$.\\
(\ref{condizione_caso_generale}) reads as
\[
2^{10n}\leq\mathbf{V}_5(\text{Mat}(\textbf{n},\textbf{m},\mathbb{F}_q))=\mathbf{V}_4(\text{Mat}(\textbf{n},\textbf{m},\mathbb{F}_q))+\mathbf{V}(S_5)\]\[\leq1+ 3\cdot2^{4n}+\frac{7}{3}2^{6n-2}+\mathbf{V}(S_4)\leq1+ 3\cdot2^{4n}+\frac{7}{3}2^{6n-2}+\frac{22}{3}2^{8n-4}+\mathbf{V}(S_5).
\]
Using (\ref{S_5}) and (\ref{disuguaglianza_a}), we have
\[
\begin{aligned}
\mathbf{V}(S_5)&=\frac{2048}{9765}(2^n-1)^2(2^{n-1}-1)^2(2^{n-2}-1)^2(2^{n-3}-1)^2(2^{n-4}-1)^2+\\&+\frac{128}{315}(2^n-1)^2(2^n-1)^2(2^{n-1}-1)^2(2^{n-2}-1)^2(2^{n-3}-1)^2+\\&+\frac{32}{63}(2^n-1)^2(2^{n-1}-1)^2(2^n-1)^2(2^{n-1}-1)^2(2^{n-2}-1)^2\leq\\&\leq2^{10n-20}+2^{20n-12}+2^{10n-8}\leq3\cdot2^{10n-8}.
\end{aligned}
\]
Hence, we obtain the following condition 
\[
2^{10n}-1\leq3\cdot2^{4n}+\frac{7}{3}2^{6n-2}+\frac{22}{3}2^{8n-4}+3\cdot2^{10n-8}.
\]
This condition fails for $n\geq1$.
\end{proof}

\section{The case \texorpdfstring{$\mathbb{F}_3^{n\times n}\times \mathbb{F}_3^{n\times n}$}{F3(n×n) × F3(n×n)}}\label{app:q=32blocks}

As before, we can derive the following.

\begin{proposition}
Let $q=3$, $t=2$, $n_1=n_2=n$, $m_1=m_2=n$ be integers and Mat(\textbf{n},\textbf{m},$\mathbb{F}_q$)=$\mathbb{F}_3^{n\times n}\oplus\mathbb{F}_3^{n\times n}$. Then
\begin{equation}\label{S_2_3}
\textbf{V}_2(\text{Mat(\textbf{n},\textbf{m},$\mathbb{F}_q$)})=1+(3^n-1)^2+\frac{1}{4}(3^n-1)^4+\frac{3}{8}(3^n-1)^2(3^{n-1}-1)^2
\end{equation}
\begin{equation}\label{S_3_3}
\textbf{V}(S_3)=\frac{54}{416}(3^n-1)^2(3^{n-1}-1)^2(3^{n-2}-1)^2+\frac{3}{16}(3^n-1)^4(3^{n-1}-1)^2
\end{equation}
\begin{equation}\label{S_4_3}
	\begin{aligned}
\textbf{V}(S_4)&=\frac{1458}{33280}(3^n-1)^2(3^{n-1}-1)^2(3^{n-2}-1)^2(3^{n-3}-1)^2+\\&+\frac{27}{416}(3^n-1)^4(3^{n-1}-1)^2(3^{n-2}-1)^2+\\&+\frac{81}{16}(3^n-1)^4(3^{n-1}-1)^4
\end{aligned}
\end{equation}
\end{proposition}

As for the case $q=2$, we use the above proposition and the Singleton-like bound for completing the proof of Proposition \ref{prop:q=2,3nonexistencecomplete} for the case $q=3$.

\begin{theorem}
Let $q=3$, $t=2$, $n_1=n_2=n$, $m_1=m_2=n$ be integers and Mat(\textbf{n},\textbf{m},$\mathbb{F}_q$)=$\mathbb{F}_3^{n\times n}\oplus\mathbb{F}_3^{n\times n}$. If $n\geq3$ and $\lfloor\frac{d-1}{2}\rfloor\in\{2,3,4\}$, then there do not exist perfect codes in Mat(\textbf{n},\textbf{m},$\mathbb{F}_q$) with minimum distance $d$.
\end{theorem}
\begin{proof}
Suppose $\mathcal{C}$ is a perfect code in Mat(\textbf{n},\textbf{m},$\mathbb{F}_q$) with minimum distance $d$, then
\begin{equation}\label{casogenerale3}
	q^{2kn}\leq\textbf{V}_k(\text{Mat(\textbf{n},\textbf{m},$\mathbb{F}_q$)})
\end{equation}
where $k=\lfloor\frac{d-1}{2}\rfloor$.\\

\textbf{Case 1:} $k=1$.\\
By using (5) we have that
\[
\mathbf{V}_1(\text{Mat}(\textbf{n},\textbf{m},\mathbb{F}_q))=1+2\frac{(3^n-1)^2}2=1+(3^n-1)^2=3^{2n}-2\cdot 3^n+2=3^n(3^n-2)+2.
\]
Since the code is perfect then 
\[ |\C| \mathbf{V}_1(\text{Mat}(\textbf{n},\textbf{m},\mathbb{F}_q))=3^{2n^2},  \]
i.e. 
\[ |\C|=\frac{2^{2n^2}}{ \mathbf{V}_1(\text{Mat}(\textbf{n},\textbf{m},\mathbb{F}_q))}=\frac{3^{2n^2}}{3^n(3^n-2)+2},  \]
which turns out to be not an integer, as the denominator is not a multiple of three, a contradiction.

\textbf{Case 2:} $k=2$.\\
(\ref{casogenerale3}) reads as
\[
3^{4n}-1\leq(3^n-1)^2+\frac{1}{4}(3^n-1)^4+\frac{3}{8}(3^n-1)^2(3^{n-1}-1)^2
\]
this condition fails for $n>0$.\\

\textbf{Case 3:} $k=3$.\\
From
\begin{equation}\label{dis_3}
	(3^a-1)^2<(3^a)^2=3^{2a} \ \ \forall a>-\frac{\lg(2)}{\lg(3)}
\end{equation}
it follows that
\[
\textbf{V}_2(\text{Mat(\textbf{n},\textbf{m},$\mathbb{F}_q$)})\leq1+3^{2n}+\frac{1}{4}3^{4n}+\frac{3}{8}3^{4n-2}\leq1+3^{4n}\left(1+\frac{1}{4}+\frac{3}{8}\right)=1+\frac{13}{8}3^{4n}
\]
and
\[
\textbf{V}(S_3)=\frac{54}{416}(3^n-1)^2(3^{n-1}-1)^2(3^{n-2}-1)^2+\frac{3}{16}(3^n-1)^4(3^{n-1}-1)^2\leq2\cdot3^{6n-2}.
\]
Using (\ref{casogenerale3}), we obtain the following condition
\[
3^{6n}\leq\textbf{V}_3(\text{Mat(\textbf{n},\textbf{m},$\mathbb{F}_q$)})= \textbf{V}_2(\text{Mat(\textbf{n},\textbf{m},$\mathbb{F}_q$)})+\textbf{V}(S_3)\leq1+\frac{13}{8}3^{4n}+2\cdot3^{6n-2},
\]
i.e.
\[
3^{6n}-1\leq \frac{13}{8}3^{4n}+2\cdot3^{6n-2},
\]
which yields a contradiction when $n\geq1$.\\

\textbf{Case 4:} $k=4$.\\
Combining (\ref{dis_3}) and (\ref{S_4_3}), we obtain an upper bound on $\textbf{V}(S_4)$:
\[
\textbf{V}(S_4)\leq 3^{8n-8}+3^{8n-6}+\frac{81}{16}3^{8n-4}\leq \frac{113}{16} \cdot3^{8n-4},
\]
(\ref{casogenerale3}) becomes
\[
3^{8n}\leq\textbf{V}_4(\text{Mat(\textbf{n},\textbf{m},$\mathbb{F}_q$)})=\textbf{V}_3(\text{Mat(\textbf{n},\textbf{m},$\mathbb{F}_q$)})+\textbf{V}(S_4)\leq1+\frac{13}{8}3^{4n}+2\cdot3^{6n-2}+\frac{113}{16} \cdot3^{8n-4}.
\]
This condition fails for $n\geq3$.
\end{proof}

\section*{Acknowledgments}
This research was partially supported by the Italian National Group for Algebraic and Geometric Structures and their Applications (GNSAGA - INdAM).

\bibliographystyle{abbrv}
\bibliography{biblio}

\end{document}